\documentclass[reqno]{amsart}
\usepackage[a4paper,margin=2.5cm,includefoot]{geometry}
\usepackage{fancyhdr}

\usepackage[titletoc,title]{appendix}

\usepackage[bb=libus]{mathalpha}
\usepackage{geometry}
\usepackage{yfonts}
\usepackage[T1]{fontenc} 
\usepackage[utf8]{inputenc} 
\usepackage{mathtools}
\usepackage{amssymb}
\usepackage{amsthm}
\usepackage{bbm}
\usepackage{mathrsfs}
\usepackage[english]{babel} 
\usepackage{enumitem}
\usepackage{dsfont}
\usepackage{url}
\usepackage{graphicx}
\usepackage{braket}
\usepackage{nicematrix}
\usepackage{cancel}

\usepackage{hyperref}
\hypersetup{linktocpage=true, colorlinks=true, breaklinks=true, linkcolor=blue, menucolor=black, urlcolor=black,citecolor=blue,filecolor=green}

\newcommand{\N}{\mathbb{N}} 
\newcommand{\R}{\mathbb{R}} 
\newcommand{\C}{\mathbb{C}} 

\newcommand{\1}{\mathds{1}}

\theoremstyle{plain}
\newtheorem{thrm}{Theorem}
\newtheorem{prop}{Proposition}
\newtheorem{cor}[prop]{Corollary}
\newtheorem{lemma}[prop]{Lemma}

\theoremstyle{definition}

\newtheorem{assumption}{Assumption}

\begin{document}
	
	\author{Steffen Polzer}
	\address{Steffen Polzer \hfill\newline
		\indent Section de mathématiques, Université de Genève
	}
	\email{steffen.polzer@unige.ch}

	\title{Lower bound on the energy-momentum relation of the polaron}
	
	\begin{abstract}
		For a class of polaron-type models, we establish a lower bound on the energy-momentum relation in terms of the vacuum overlap and the spectral gap of the total momentum zero Hamiltonian.
		We show convergence of the rescaled mean square displacement of the associated path measure to the inverse of the effective mass. We derive a probabilistic criterion for the absence of ground states at large total momentum.
	\end{abstract}

	\maketitle
	
	\section{Introduction and main results}
	The polaron describes the interaction of a charged particle with a bosonic field. For suitable functions $\omega$ and $v$ on $\R^d$, the Hamiltonian $H(P)$ at fixed total momentum $P\in \R^d$ between particle and field is the self-adjoint operator on bosonic Fock space 
  	$\mathcal F(L^2(\R^d))$ 
	over $L^2(\R^d)$ formally given by
	\begin{equation}
		\label{Equation: Hamiltonianof Polaron at fixed total momentum}
		H(P) = \frac12 |P - P_f|^2 + \mathrm d\Gamma(\omega) + \phi(v)
	\end{equation}
	where $P_{\mathrm f} = \mathrm d\Gamma(\operatorname{id}_{\R^d})$ and $\mathrm d\Gamma(\omega)$ are field momentum and energy respectively (with $\operatorname{id}_{\R^d}$ and $\omega$ acting as multiplication operators) and $\phi(v) = a^*(v) + a(v)$ is the Segal field operator, with $a^*(v)$ and $a(v)$ being the creation and annihilation operators associated to $v$.
	We will study the energy-momentum relation
	\begin{equation*}
		E(P) \coloneqq \inf \sigma(H(P))
	\end{equation*}
	i.e. the ground state energy as a function of $P$.
	For the Fröhlich model of the polaron, corresponding to  $d = 3$, $\omega(k) = 1$ and $v(k) = \sqrt{\alpha/(2\pi^2)} |k|^{-1}$ it was shown in \cite{LampartMitrouskasMysliwy.2023} that $E$ has a strict global minimum in the origin. In \cite{Polzer.2023}, we improved this by showing that $E$ is an increasing and concave function of $|P|^2$, both so strictly on $\{|P|^2:\, P\in \mathcal I_0\}$ where
	\begin{equation*}
		\mathcal I_0 \coloneqq \big\{P \in \R^d: \, E(P) <E_{\text{ess}}(P) \big \}
	\end{equation*}
	with $E_{\text{ess}}(P)$ being the bottom of the essential spectrum of $H(P)$. This has recently been generalized by Desio and Seiringer to other polaron-type models \cite{Desio.Seiringer.2025}.
	In the following, we establish a strictly positive lower bound on $E(P) - E(P_0)$ for $P_0 \in \mathcal I_0$ and $P\in \R^d$ with $|P| > |P_0|$ in terms of the vacuum overlap and the spectral gap of $H(P_0)$.
	
	To state our precise assumptions on $v$ and $\omega$, 
	we recall that a function $f \in C^\infty(0, \infty)$ is called completely monotone if $(-1)^n f^{(n)}(t) \geq 0$ for all $t> 0$ and $n\in \N$. By Bernstein's theorem, a function $f$ is completely monotone if and only if it is the Laplace transform of a measure on $[0, \infty)$ \cite{Donoghue.1974}. A non-negative function $f \in C^\infty(0, \infty)$ is called a Bernstein function if $f'$ is completely monotone, in which case $e^{-tf}$ is completely monotone for all $t\geq 0$. We will make the following standing assumption.
	
	\begin{assumption}
		\label{Assumption: Assumption 1}
		The functions $\omega$ and $v$ are radially symmetric. The function $|k|^2 \mapsto |v(k)|^2$ is completely monotone with $v \not \equiv 0$ and $\lim_{k\to \infty} v(k) = 0$. The function $|k|^2 \mapsto  \omega(k)$ is strictly positive on $(0, \infty)$ and a Bernstein function. Moreover, we have
		\begin{equation}
		\label{Equation: integrability Assumption}
				\int_{\R^d} \frac{|v(k)|^2}{1+|k|^2} \frac{1+\omega(k)}{\omega(k)} \, \mathrm dk < \infty.
		\end{equation}
		
	\end{assumption}

	As a consequence of \eqref{Equation: integrability Assumption}, $\phi(v)$ is infinitesimally form-bounded with respect to $\frac12 |P - P_f|^2 + \mathrm d\Gamma(\omega)$ and \eqref{Equation: Hamiltonianof Polaron at fixed total momentum} defines a self-adjoint operator via its quadratic form \cite{Desio.Seiringer.2025}.
	Assumption \ref{Assumption: Assumption 1} is in particular satisfied for the Fröhlich polaron, i.e.\ for $d = 3$, $\omega(k) = 1$ and $v(k) = \sqrt{\alpha/(2\pi^2)} |k|^{-1}$ for some $\alpha>0$ and for the Nelson model with a Gaussian ultraviolet regularization,  for which $\omega(k) = \sqrt{|k|^2 + m^2}$ and $v(k) = \alpha  e^{-\varepsilon |k|^2} \omega(k)^{-1/2}$ for some $m\geq 0, \, \varepsilon>0, \, \alpha>0$.

	For $P\in \R^d$, we denote by $\rho(P)$ and $\Delta(P)$ the vacuum overlap and spectral gap of $H(P)$ defined by
	\begin{align*}
		\rho(P) \coloneqq \big\langle \Omega, \1_{ \{E(P)\}}(H(P)) \Omega \big\rangle, \quad \Delta(P) \coloneqq \inf \big\{ \lambda >0:\, E(P) + \lambda \in \sigma(H(P)) \big\}
	\end{align*}
	where $\1_{ \{E(P)\}}(H(P))$ is the orthogonal projection onto $\operatorname{ker}(H(P) - E(P))$ and where $\Omega$ is the Fock vacuum.  By a standard application of Perron--Frobenius theory, if there exists a ground state of $H(P)$, i.e.\ an eigenvector to the eigenvalue $E(P)$, then it is unique up to a phase and non-orthogonal to $\Omega$, see  Appendix \ref{Appendix: positivity improving} for details. Hence, $H(P)$ has a ground state if any only if $\rho(P)>0$.
	In particular, we have $\rho(P) >0$ for $P\in \mathcal I_0$.
	Also notice that $\Delta(P) > 0$ holds for $P\in \mathcal I_0$ by definition of the essential spectrum.
	\begin{thrm}
		\label{Theorem: Lower bound on EMR}
		We have for all $P, P_0\in \R^d$ with $|P| \geq |P_0|$
		\begin{equation*}
			E(P) \geq E(P_0) + \tfrac{1}{2}\big(\Delta(P_0) + \delta(P, P_0)\big) - \sqrt{\tfrac{1}{4}\big(\Delta(P_0) + \delta(P, P_0)\big)^2 - \Delta(P_0) \rho(P_0) \delta(P, P_0)}
		\end{equation*}
		where  $\delta(P, P_0) \coloneqq \frac{1}{2}\big(|P|^2 - |P_0|^2\big)$.
	\end{thrm}
	
	For the Fröhlich polaron, $E_{\text{ess}}(P) = E(0) + 1$ for all $P$ \cite{Sp88} and it has been conjectured that $\mathcal I_0$ is bounded, i.e.\ that $E(P) = E(0) + 1$ for sufficiently large $P$ \cite{Feynman.1972}. This, however, has so far only been shown at small coupling \cite{Da17}. While the lower bounds on $E(P)$ given in Theorem \ref{Theorem: Lower bound on EMR} might be of interest in this context, we point out that
		\begin{align*}
			\lim_{P \to \infty} \tfrac{1}{2}\big(\Delta(P_0) + \delta(P, P_0)\big) - \sqrt{\tfrac{1}{4}\big(\Delta(P_0) + \delta(P, P_0)\big)^2 - \Delta(P) \rho(P_0) \delta(P, P_0)} &= \Delta(P_0) \rho(P_0)
		\end{align*}
	is, for $P\in \mathcal I_0$, strictly smaller than $E_{\operatorname{ess}}(P_0) - E(P_0)$. We also notice that for the Fröhlich polaron $\Delta(0) =  E_{\operatorname{ess}}(0) - E(0)$ at sufficiently weak coupling \cite{Seiringer.2023} whereas $\Delta(P) <  E_{\operatorname{ess}}(P) - E(P)$ for fixed $P$ at sufficiently large coupling \cite{MitrouskasSeiringer.23}.
	
	For the proof of Theorem \ref{Theorem: Lower bound on EMR}, we will use the representation
	\begin{equation}
		\label{Equation: E from semigroup}
		E(P) = -\lim_{T \to \infty} \frac{1}{T} \log \big\langle \Omega, e^{-T H(P)} \Omega \big\rangle
	\end{equation}
	of the ground state energy which follows from Perron--Frobenius theory, see Appendix \ref{Appendix: positivity improving} below. For the Fröhlich polaron, we applied in \cite{Polzer.2023} a point process representation of $\langle \Omega, e^{-TH(P)} \Omega \rangle$ by expanding the exponential in the Feynman--Kac formula into a series, based on the representation of the associated path measure developed by Mukherjee and Varadhan in \cite{MukherjeeVaradhan.2020}.
	For the proof of Theorem \ref{Theorem: Lower bound on EMR}, we will derive a similar point process representation under the more general assumptions given above. However, in order to avoid the difficulties that arise while dealing with the Feynman--Kac formula when $v$ is not square integrable, we will take a different approach and instead start from the Dyson-type expansion derived in \cite{Desio.Seiringer.2025}.
	This approach will additionally allow us to prove a Feynman--Kac formula for $\langle \Omega, e^{-TH(P)} \Omega \rangle$ that remains valid even if $v\notin L^2(\R^d)$. To state the latter, notice that the function
	\begin{equation*}
		|k|^2 \mapsto |v(k)|^2 e^{-t \omega(k)}
	\end{equation*}
	is under our assumptions on $\omega$ and $v$ for every $t\geq 0$ completely monotone. By Bernstein's theorem, there hence exists for every $t\geq 0$ a locally finite measure $\tilde \kappa(t,\cdot)$ on $(0, \infty)$ 
	such that for all $k\in \R^d$
	\begin{equation}
		\label{Equation: Definition of kappatilde}
		|v(k)|^2 e^{-t\omega(k)} = \int_{(0, \infty)} e^{- u |k|^2/2} \, \tilde \kappa(t, \mathrm du).
	\end{equation}
	We denote by $\kappa(t, \cdot)$ the unique measure such that for every non-negative measurable function $f$ on $(0, \infty)$
	\begin{equation*}
		\int_{(0, \infty)}f(u) \, \kappa(t, \mathrm du) = \int_{(0, \infty)} (2\pi/u)^{d/2} f(u^{-1/2}) \, \tilde \kappa(t, \mathrm du)
	\end{equation*}
	i.e.\ $\kappa(t, \cdot)$ has Radon--Nikodym derivative $u \mapsto (2\pi)^{d/2} u^d$ with respect to the pushforward of $\tilde \kappa(t, \cdot)$ under the map $u\mapsto u^{-1/2}$. We define
	
	\begin{equation*}
		w(t, x) \coloneqq \int_{(0, \infty)} e^{-u^2 |x|^2/2} \, \kappa(t, \mathrm du), \quad t\geq 0, \, x\in \R^d.
	\end{equation*}

	\begin{thrm}
		\label{Theorem: Feyman--Kac formula}
		For all $T\geq 0$ and $P\in \R^d$
		\begin{equation}
		\label{Equation: FK}
		\big\langle \Omega, e^{-TH(P)} \Omega \big\rangle = \mathbb E\bigg[e^{\mathrm i P \cdot X_{0, T}}  \operatorname{exp}\Big(\frac{1}{2}\int_{[0, T]^2} w(|t-s|, X_{s, t}) \, \mathrm ds \mathrm dt   \Big)  \bigg]
		\end{equation}
		where $(X_t)_{t \geq 0}$ is a Brownian motion on $\R^d$ and where $X_{s, t} \coloneqq X_t - X_s$ for $s, t\geq 0$.
	\end{thrm}
	Provided that $v\in L^2(\R^d)$, we will see in the proof of Theorem \ref{Theorem: Feyman--Kac formula} that
		\begin{equation}
			\label{Equation: definition of w for UV regular interactions}
			w(t, x) = \int_{\R^d} |v(k)|^2 e^{-\omega(k)t}e^{\mathrm i k \cdot x} \, \mathrm dk
		\end{equation}
		such that \eqref{Equation: FK} reduces to the usual Feynman--Kac formula \cite{DySp20}. For the Fröhlich polaron, we have $\tilde \kappa(t, \mathrm du) = (4\pi^2)^{-1}\alpha e^{-t} \, \mathrm du$ such that $\kappa(t, \mathrm du) = \sqrt{2/\pi} \alpha e^{-t} \, \mathrm du$ yielding $w(t, x) = \alpha e^{-t} |x|^{-1}$. Hence, \eqref{Equation: FK} reduces once more to the known Feynman--Kac formula \cite{DySp20}, which is commonly obtained by first introducing a suitable ultraviolet regularization of $v$ and then taking the limit in \eqref{Equation: definition of w for UV regular interactions}.
		
	Considering Theorem \ref{Theorem: Feyman--Kac formula}, let us introduce the path measure
	\begin{equation}
		\label{Equation: Definition of hat P}
		\widehat{\mathbb P}_{T}(\mathrm dX) \coloneqq \frac{1}{Z_{T}} \operatorname{exp}\Big(\frac{1}{2}\int_{[0, T]^2} w(t-s, X_{s, t}) \mathrm ds \mathrm dt   \Big) \mathcal W(\mathrm dX)
	\end{equation}
	where $\mathcal W$ is the Wiener measure, i.e.\ the distribution of Brownian motion, and where $Z_{T} = \langle \Omega, e^{-TH(0)} \Omega \rangle$ acts as a normalization constant. 
	By Theorem \ref{Theorem: Feyman--Kac formula} we have
	\begin{equation}
		\label{Equation: fraction of diagonal elements is char fct}
			\frac{\langle \Omega, e^{-TH(P)} \Omega \rangle}{\langle \Omega, e^{-TH(0)} \Omega \rangle} = \widehat{\mathbb E}_T[e^{\mathrm i P \cdot X_{0, T}}]
	\end{equation}
	where the expected value $\widehat{\mathbb E}_T$ is taken with respect to the probability measure $\widehat{\mathbb P}_T$. For the Fröhlich model, the path measure $\widehat{\mathbb P}_{T}$ has in recent years been used to study the effective mass $m_{\operatorname{eff}}$ defined by
	\begin{equation}
		\label{Equation: Definition of effective mass}
		m_{\operatorname{eff}}^{-1} \coloneqq 2\lim_{P\downarrow 0} \frac{E(P) - E(0)}{|P|^2}. 
	\end{equation}
	In \cite{MukherjeeVaradhan.2020, MukherjeeVaradhan.2022}, it was shown that the path measure $\widehat{\mathbb P}_T$ of the Fröhlich polaron satisfies a central limit theorem under diffusive rescaling: there exists some $\sigma^2>0$ such that the distribution of $T^{-1/2} X_{0, T}$ converges to $\mathcal N(0, \sigma^2 I_{3})$ (i.e.\ the distribution of a centered Gaussian vector with covariance matrix $\sigma^2 I_{3}$) as $T\to \infty$, see also \cite{BetzSpohn.2005, Gubinelli2006, Mukherjee.2022, BetzPolzer.2022} for related central limit theorems. 
	Using the analyticity of $P \mapsto E(P)$ and $P\mapsto \1_{\{E(P)\}}(H(P))$ on $\mathcal I_0$ and the representation of the characteristic function of $X_{0, T}$ given in \eqref{Equation: fraction of diagonal elements is char fct}, one can then show that the effective mass can be expressed as
 	$m_{\operatorname{eff}} = (\sigma^2)^{-1}$ \cite{DySp20}. As it becomes apparent from the proof of the central limit theorem, the limiting variance $\sigma^2$ coincides with the limit of the rescaled mean square displacement $\hat \sigma^2 \coloneqq \lim_{T \to \infty} \hat \sigma_T^2(0)$ where
	\begin{equation*}
	 \hat \sigma_T^2(0) \coloneqq \frac{1}{dT} \widehat{\mathbb E}_T\big[|X_{0, T}|^2]
	\end{equation*}
	for $T>0$.
	This has then been applied in \cite{BetzPolzer.2023, Se24, BazaesMukherjeeSellkeVaradhan.2023} in order to derive lower bounds on the effective mass by deriving upper bounds on the mean square displacement, see also \cite{BetzSchmidtSellke.2025a, BetzSchmidtSellke.2025b} for recent estimates on the mean square displacement for other polaron-type measures. However, it does not seem to be known whether the equality $m_{\text{eff}}^{-1} = \hat \sigma^2$ holds for general polaron-type models, see the discussion in \cite[Section 1.2]{BetzSchmidtSellke.2025b}. For the models covered by our assumptions, we will in particular give an affirmative answer to the question raised in \cite{BetzSchmidtSellke.2025b} of whether sub-diffusivity of the path measure i.e.\ $\hat \sigma^2 = 0$ implies an infinite effective mass.
	
	Our approach in fact allows for more generality. Let $\mathcal E:[0, \infty) \to \R$ be the unique function such that $E(P) = \mathcal E(|P|^2)$ for all $P\in \R^d$. We will see in the proof of Theorem \ref{Theorem: Convergence of means square displacement} below that $\mathcal E$ is concave and increasing, compare also to \cite[Theorem 1]{Polzer.2023}, \cite[Corollary 1]{Desio.Seiringer.2025}. In particular the left and right derivatives $\partial^- \mathcal E$ and $\partial^+ \mathcal E$ of $\mathcal E$ exist on $(0, \infty)$ and $[0, \infty)$ respectively and
	\begin{equation*}
		 m_{\operatorname{eff}}^{-1} = 2 \cdot \partial^+ \mathcal E(0)
	\end{equation*}
	exists.
	We denote by $\operatorname{sinc}(x) = \sin(x)/x$ the sinus cardinalis (with $\operatorname{sinc}(0) \coloneqq 1$) and define for all $T>0$ and $P\in \R^d\setminus \{0\}$
	\begin{equation*}
		\hat \sigma^2_T(P) \coloneqq
			\frac{1}{T} \frac{ \widehat{\mathbb E}_T\big[\big(\tfrac{P}{|P|}\cdot X_{0, T}\big)^2 \operatorname{sinc}(P \cdot X_{0, T})\big]}{\widehat{\mathbb E}_T[\cos(P\cdot X_{0, T})]}.
	\end{equation*}
	
	We call $\rho$ right-continuous at $P_0$ if $|P| \mapsto \rho(P)$ is right-continuous at $|P_0|$. That is, we call $\rho$ right-continuous at $P\in \R^d\setminus \{0\}$ if $\lim_{\lambda \downarrow 1} \rho(\lambda P) = \rho(P)$ holds and right-continuous at $P=0$ if  $\lim_{P \to 0} \rho(P) = \rho(0)$ holds.
	\begin{thrm}
		\label{Theorem: Convergence of means square displacement}
		The following holds for all $P\in \R^d$.
		\begin{enumerate}
			\item 
		We have
		\begin{equation*}
			 2 \cdot \partial^+ \mathcal E(|P|^2) \leq \liminf_{T\to \infty} \hat \sigma^2_T(P).
		\end{equation*}
		If $P\neq 0$, we additionally have
		\begin{equation*}
			2 \cdot \partial^- \mathcal E(|P|^2) \geq \limsup_{T\to \infty} \hat \sigma^2_T(P).
		\end{equation*}

		\item If $\rho$ is right-continuous at $P$ and if $\rho(P)>0$ then 
		\begin{equation*}
			2 \cdot \partial^+ \mathcal E(|P|^2) = \lim_{T\to \infty} \hat \sigma^2_T(P).
		\end{equation*}
		\item We have
		\begin{equation*}
			\rho(P) \leq \liminf_{T\to \infty} \hat \sigma^2_T(P) \leq  \limsup_{T\to \infty} \hat \sigma^2_T(P) \leq 1.
		\end{equation*}
	\end{enumerate}
	\end{thrm}
	As mentioned above, for the Fröhlich polaron we have $E_{\text{ess}}(P) = E(0)+1$ for all $P\in \R^d$ and we easily obtain the following necessary and sufficient criterion for the conjectured boundedness of $\mathcal I_0$. 
	
	\begin{cor}
		For the Fröhlich polaron, $\mathcal I_0$ is bounded if and only if there exists some $P\in \R^d$ such that
		\begin{equation*}
			\liminf_{T\to \infty} \hat \sigma_T^2(P) = 0.
		\end{equation*}
	\end{cor}

	\subsection*{Acknowledgments} The author acknowledges funding from the Swiss State Secretariat for Education, Research and Innovation (SERI) through
	the consolidator grant ProbQuant, and funding from the Swiss National Science Foundation through the
	NCCR SwissMAP grant.

	\section{The renewal transform under rank-one perturbations}
	\label{Section: The renewal transform under rank one perturbations}
	For the proof of Theorem \ref{Theorem: Lower bound on EMR}, we will use the concept of the renewal transform of a probability measure which we recently introduced in \cite{HinrichsPolzer.2025}. Let $\mu$ be a probability measure on $\R$ with lower bounded support and finite mean $m\coloneqq\int_\R x\,\mu(\mathrm d x)$. We denote by $(Z_t)_{t\geq 0}$ its Laplace transform, i.e.\
	\begin{equation*}
		Z_t \coloneqq \int_{[E, \infty)} e^{-tx} \, \mu(\mathrm dx)
	\end{equation*}
	for all $t\geq 0$, where $E \coloneqq \inf \operatorname{supp}(\mu)$ denotes the bottom of the support of $\mu$.
	Let $\mathbf P$ be a probability measure on the space
	\begin{equation*}
		\mathcal D \coloneqq \{y:[0, \infty) \to \{0, 1\}:\, y \text{ is càdlàg and }y_0 = 0\}
	\end{equation*}
	which we equip, as usual, with the $\sigma$-algebra generated by the evaluation maps $y\mapsto Y_t(y) \coloneqq y_t, \ t\geq 0$. We denote by $Y \coloneqq (Y_t)_{t\geq 0}$ the canonical stochastic process with law $\mathbf P$. We say $Y$ is dormant at time $t$, if $Y_t=0$ and active at time $t$ if $Y_t = 1$. We denote by
	\begin{equation*}
		d_1 \coloneqq \inf\{t\geq 0: Y_t = 1\}, \quad a_1 \coloneqq 
		\inf\{t-d_1: t\geq d_1,\, Y_t = 0\}
	\end{equation*}
	the first dormant and the first active period (which might be infinite) and define $T_1 \coloneqq d_1 + a_1$ to be the first return to being dormant. We call $Y$ an alternating renewal process (with respect to $\mathbf P$) if either $T_1 = \infty$ almost surely or if $(Y_{T_1+t})_{t\geq 0}$ is conditionally on the event $\{T_1< \infty\}$ independent of $(d_1, a_1)$ and has law $\mathbf P$. Provided that $t\geq 0$ is such that $\mathbf P(Y_t = 0)>0$, we define $\mathbf P_t$ to be the law of $(Y_s)_{0 \leq s \leq t}$ under the probability measure $\mathbf P(\, \cdot \,|Y_t = 0)$ and denote by $\mathbf E_t$ the expected value taken with respect to $\mathbf P_t$. 
	To simplify notation, we denote by $\operatorname{Exp}(0)$ the law of a random variable which is almost surely $+\infty$.
	\begin{thrm}[{\cite[Theorem 4.1]{HinrichsPolzer.2025}}]
		\label{Theorem: Existence renewal transform}
		There exists a unique probability measure $\mathbf P$ on $\mathcal D$ such that $Y$ is an alternating renewal process, such that $d_1$ is exponentially distributed and independent of $a_1$ and such that
		\begin{equation}
			\label{Equation: Representation of Z in terms of renewal transform}
			e^{Et} Z_t = \mathbf P(Y_t = 0)
		\end{equation}
		for all $t\geq 0$. We have $d_1 \sim \operatorname{Exp}(m-E)$ under $\mathbf P$ and for all $0= t_0 \leq t_1 \leq \hdots \leq t_n = t$
		\begin{equation}
			\label{Equation: fractions of Z in terms of renewal transform}
			\frac{1}{Z_t}\prod_{i=0}^{n-1} Z_{t_{i+1}-t_i} = \mathbf P_t\big(Y_{t_1} = Y_{t_2} = \hdots = Y_{t_{n-1}} = 0\big).
		\end{equation}		 
	\end{thrm}
	We call the probability measure $\mathbf P$ the {\em renewal transform} of $\mu$ and the family of measures $(\mathbf P_t)_{t\geq 0}$ the {\em finite volume renewal transforms} of $\mu$.

	Let us now assume that $\mu$ is the spectral measure of some lower bounded self-adjoint operator $A$ acting on some Hilbert space $\mathcal H$ with respect to some unit vector $\psi \in \operatorname{dom}(|A|^{1/2})$. For $\alpha \in \R$, let $\mu_\alpha$ be the spectral measure of
	\begin{equation*}
		A_\alpha \coloneqq A + \alpha \langle \psi, \, \cdot \, \rangle \psi
	\end{equation*}
	with respect to $\psi$ and let $Z_{\alpha, t} \coloneqq \int_\R e^{-tx} \, \mu_\alpha(\mathrm dx)$ be its Laplace transform. We will now determine the finite volume renewal transforms $\mathbf P_{\alpha, t}$ of $\mu_\alpha$ in terms of the finite volume renewal transform $\mathbf P_{t}$ of $\mu$. For $t\geq 0$, we define
	\begin{equation*}
		D_t \coloneqq \int_0^t (1-Y_s) \, \mathrm ds
	\end{equation*}
	to be the total dormant time in the interval $[0, t]$.
	
	\begin{prop}
		\label{Proposition: Renewal transform under rank-one perturbations}
		For $t\geq 0$ the finite volume renewal transform $\mathbf P_{\alpha, t}$ of $\mu_\alpha$ is given by
		\begin{equation*}
			\mathbf P_{\alpha, t}(\mathrm dY) = \frac{1}{\mathbf Z_{\alpha, t}} e^{-\alpha D_t(Y)} \mathbf P_{t}(\mathrm dY)
		\end{equation*}
		where $\mathbf Z_{\alpha, t}= Z_{\alpha, t}/Z_t$ acts as a normalization constant.
	\end{prop}
	
	\begin{proof}
		We define the functions $f, \tilde f: \R \times [0, \infty) \to [0, \infty)$ by
		\begin{equation}
			\label{Equation: partial derivative of f wrt alpha}
			f(\alpha, t) \coloneqq e^{Et} Z_{\alpha, t}, \quad \tilde f(\alpha, t) \coloneqq \mathbf E\big[e^{-\alpha D_t} \1_{\{Y_t = 0\}}\big]
		\end{equation}
		where the expected value $\mathbf E$ is taken with respect to the renewal transform $\mathbf P$ of $\mu$. Expanding $e^{-tA_\alpha}$ into a Dyson series yields
		\begin{equation}
			\label{Equation: convolution eq for f}
			\partial_\alpha f(\alpha, t) =  e^{E t}\partial_\alpha \big\langle \psi, e^{-tA_\alpha} \psi \big\rangle = - e^{E t} \int_0^t Z_{\alpha, s} Z_{\alpha, t-s} \, \mathrm ds = - \int_0^t f(\alpha, s) f(\alpha, t-s) \, \mathrm ds.
		\end{equation}
		Moreover, we have by dominated convergence and Fubini's theorem
		\begin{equation}
			\label{Equation: partial derivative of tilde f wrt alpha}
			\partial_\alpha \tilde f(\alpha, t) = -\mathbf E\big[D_t e^{-\alpha D_t} \1_{\{Y_t = 0\}}\big] = -\int_0^t \mathbf E\big[e^{-\alpha D_t}  \1_{\{Y_s = Y_t = 0\}} \big] \, \mathrm ds.
		\end{equation}
		By the memorylessness property of the exponential distribution\footnote{Alternatively, one can show this using the Poisson point process of the $M/G/\infty$-queue which we used in the proof of \cite[Theorem 4.1]{HinrichsPolzer.2025} to explicitly construct $\mathbf P$.}, the processes $(Y_r)_{0\leq r \leq s}$ and $(Y_r)_{s\leq r \leq t}$ are under $\mathbf P(\, \cdot \, |Y_s = 0)$ independent, and the distribution of $(Y_{s+r})_{0\leq r \leq t-s}$ under $\mathbf P(\, \cdot \, |Y_s = 0)$ coincides with the distribution of $(Y_r)_{0\leq r \leq t-s}$ under $\mathbf P$. If we define $D_{s, t} \coloneqq \int_s^t (1-Y_r) \, \mathrm dr$ to be the total dormant time in the interval $[s, t]$, we hence have 
		\begin{align}
			\label{Equation: factor expected values perturbed measure}
			\mathbf E\big[e^{-\alpha D_t}  \1_{\{Y_s = Y_t = 0\}} \big]  &= \mathbf P(Y_s = 0) \mathbf E\big[e^{-\alpha D_s}|Y_s = 0\big] \mathbf E\big[e^{-\alpha D_{s, t}} \1_{\{Y_t = 0\}}|Y_s = 0\big] \nonumber \\
			&= \mathbf E\big[e^{-\alpha D_s}\1_{\{Y_s = 0\}}\big] \mathbf E\big[e^{-\alpha D_{t-s}} \1_{\{Y_{t-s} = 0\}}\big] 
		\end{align}
		such that \eqref{Equation: partial derivative of tilde f wrt alpha} becomes
		\begin{equation}
			\label{Equation: convolution eq for tilde f}
			\partial_\alpha \tilde f(\alpha, t) = - \int_0^t \tilde f(\alpha, s) \tilde f(\alpha, t-s) \, \mathrm ds.
		\end{equation}
		Since 
		\begin{equation*}
			f(0, t) = e^{Et} Z_{t} = \mathbf P(Y_t = 0) = \tilde f(0, t) 
		\end{equation*}
		we obtain with \eqref{Equation: convolution eq for f}, \eqref{Equation: convolution eq for tilde f} and Lemma \ref{Lemma: convolution eq has at most one sol} below that
		\begin{equation}
			\label{Equation: eq to determine bold Z}
			e^{Et} Z_{\alpha, t} = \mathbf E\big[e^{-\alpha D_t} \1_{\{Y_t = 0\}}\big]
		\end{equation}
		holds for all $t\geq0$. 
		We define the probability measure $\widehat{\mathbf P}_{\alpha, t}$ by
		\begin{equation*}
			\widehat{\mathbf P}_{\alpha, t}(\mathrm dY) \coloneqq \frac{1}{\mathbf Z_{\alpha, t}} e^{-\alpha D_t(Y)} \mathbf P_{t}(\mathrm dY).
		\end{equation*}
		We have with \eqref{Equation: eq to determine bold Z}
		\begin{equation}
			\label{Equation: representation of bold Z in terms of Z}
			\mathbf Z_{\alpha, t} \coloneqq \mathbf E_t\big[e^{-\alpha D_t}\big] =  \frac{e^{Et} Z_{\alpha, t}}{\mathbf P(Y_t = 0)} = Z_{\alpha, t}/Z_t.
		\end{equation}
		Another application of \eqref{Equation: eq to determine bold Z} yields for all $0= t_0 \leq t_1 \leq \hdots \leq t_n = t$
		\begin{align*}
			e^{Et}\prod_{i=0}^{n-1} Z_{\alpha, t_{i+1}-t_i} &= \prod_{i=0}^{n-1} e^{E(t_{i+1}-t_i)}Z_{\alpha, t_{i+1}-t_i} \\
			&= \prod_{i=0}^{n-1}\mathbf E\big[e^{-\alpha D_{t_i, t_{i+1}}} \1_{\{Y_{t_{i+1}} = 0\}} |Y_{t_i} = 0 \big] \\
			&= \mathbf E\big[e^{-\alpha D_t}  \1_{\{Y_{t_1} = \hdots = Y_{t_{n-1}} = Y_t = 0\}} \big] 
		\end{align*}
		where the last equality follows by inductively applying the argument in \eqref{Equation: factor expected values perturbed measure}.
		With \eqref{Equation: fractions of Z in terms of renewal transform} and \eqref{Equation: representation of bold Z in terms of Z}, we hence obtain
		\begin{align*}
			\mathbf P_{\alpha, t}\big(Y_{t_1} = \hdots = Y_{t_{n-1}} = 0 \big) &= 
			\frac{1}{Z_{\alpha, t}}\prod_{i=0}^{n-1}Z_{\alpha, t_{i+1}-t_i} \\
			&= \frac{\mathbf E\big[e^{-\alpha D_t}  \1_{\{Y_{t_1} = \hdots = Y_{t_{n-1}} = Y_t = 0\}} \big]}{e^{Et}Z_{\alpha, t}} \\
			&= \frac{\mathbf P(Y_t = 0)\mathbf E_t\big[e^{-\alpha D_t}  \1_{\{Y_{t_1} = \hdots = Y_{t_{n-1}} = 0\}} \big]}{e^{Et}Z_{\alpha, t}} \\
			&= \widehat{\mathbf P}_{\alpha, t}\big(Y_{t_1} = \hdots = Y_{t_{n-1}} = 0 \big).
		\end{align*}
		Hence, $\mathbf P_{\alpha, t}$ and $\widehat{\mathbf P}_{\alpha, t}$ coincide on an intersection stable generator and are hence equal.	
	\end{proof}
	
	\begin{lemma}
		\label{Lemma: convolution eq has at most one sol}
		Let $f_0:[0, \infty) \to (0, \infty)$ be a continuous function. Then there exists at most one function $f:\R \times [0, \infty) \to (0, \infty)$ such that the following holds.
		\begin{enumerate}
			\item The function $f$ is continuously differentiable in the first and continuous in the second variable.
			\item For every $\alpha_0 \geq 0$ there exist a constant $c>0$ such that 
			\begin{equation}
				\label{Equation: at most exponential growth}
				\forall t\geq 0\, \forall \alpha \in[-\alpha_0, \alpha_0]:\, |f(\alpha, t)| \leq c e^{ ct}.
			\end{equation}
			\item The function $f$ solves
			\begin{equation}
				\label{Equation: Diff eq for f}
				\begin{cases}
					f(0, t) = f_0(t) &\text{ for all } t\geq 0\\
					\partial_\alpha f(\alpha, t) = -\int_0^t f(\alpha, t-s) f(\alpha, s) \, \mathrm ds &\text{ for all }(\alpha, t) \in \R \times [0, \infty).
				\end{cases}
			\end{equation}
		\end{enumerate}
	\end{lemma}
	
	\begin{proof}
		Fix some $\alpha_0> 0$ and let $c > 0$ be such that \eqref{Equation: at most exponential growth} holds. By \eqref{Equation: Diff eq for f}, there exists some suitable constant $C>c$ such that
		\begin{equation}
			\label{Equation: at most exponential growth for derivative}
			|\partial_\alpha f(\alpha, t)| \leq c^2 t e^{c t} \leq C e^{C t}
		\end{equation}
		for all $t\geq 0$ and $\alpha \in [-\alpha_0, \alpha_0]$. By taking the Laplace transform in \eqref{Equation: Diff eq for f} we obtain for all $\lambda > C$
		\begin{equation}
			\label{Equation: diff eq for LT}
			\partial_\alpha \mathcal L(f(\alpha, \cdot ))(\lambda) = - \mathcal L(f(\alpha, \cdot ))(\lambda)^2
		\end{equation}
		where \eqref{Equation: at most exponential growth for derivative} and the continuity of $\partial_\alpha f$ allow us the interchange the Laplace transform with the partial derivative in $\alpha$.
		Hence,
		\begin{equation*}
			\partial_\alpha \Big(\frac{1}{\mathcal L(f(\alpha, \cdot ))(\lambda)}\Big) = 1 
		\end{equation*}
		yielding after integration
		\begin{equation}
			\label{Equation: Sherman–Morrison formula}
			\mathcal L(f(\alpha, \cdot ))(\lambda) = \frac{ \mathcal L(f_0)(\lambda)}{1+ \alpha \mathcal L(f_0)(\lambda)}
		\end{equation}
		for all sufficiently large $\lambda$. The statement hence follows from Lerch's theorem, i.e.\ from the fact that a continuous function is uniquely determined by the values of its Laplace transform on some set of the form $[\tilde C, \infty)$.
	\end{proof}

	\section{The renewal transform of the spectral measures}
	For the proof of Theorem \ref{Theorem: Lower bound on EMR}, we will identify the finite volume renewal transforms of the spectral measure of $H(P)$ with respect to $\Omega$. We will then compare the latter with the finite volume renewal transforms of a family of rank-perturbations of $H(P)$. 
	Our starting  point will be the Dyson-type expansion derived in \cite{Desio.Seiringer.2025}. To state the latter, let us introduce some notation. For $n\in \mathbb N$ let $\mathcal W_{2n}$ be the set of Wick pairings of $2n$ elements, i.e.\ the set of all permutations $\pi$ of $\{1, \hdots 2n\}$ satisfying
	\begin{equation*}
		\forall j\in \{1, \hdots, n-1\}: \pi(2j-1)< \pi(2j+1), \quad \forall j\in \{1, \hdots, n\}: \pi(2j-1)< \pi(2j).
	\end{equation*} 
	We identify $\pi\in \mathcal W_{2n}$ with the pairing $\big\{\big(\pi(2j-1), \pi(2j)\big): 1\leq j \leq n\big\}$ of $\{1, \hdots, 2n\}$. For $j\in \{0, \hdots, 2n\}$ and $\pi\in \mathcal W_{2n}$ let
	\begin{equation*}
		\mathcal M_j^\pi \coloneqq \big\{1\leq i\leq n:\, \pi(2i-1) \leq j < \pi(2i)\big\}
	\end{equation*}
	 be the set of pairs whose connecting line crosses $j$. Notice that in particular $\mathcal M_0^\pi = \mathcal M_{2n}^\pi = \emptyset$. We set for $\pi\in \mathcal W_{2n}$, $j\in \{0, \hdots, 2n\}$ and $k\in \R^{dn}$
	 \begin{equation*}
	 	\mathcal E_P^{(\pi, j)}(k) \coloneqq \frac{1}{2} \Big|P -\sum_{l\in \mathcal M_j^\pi} k_l\Big|^2 + \sum_{l\in \mathcal M_j^\pi} \omega(k_l)
	 \end{equation*}
	 and define the simplex $\triangle_{2n, T} \coloneqq \{\sigma \in (0, \infty)^{2n+1}:\, \sum_{i=0}^{2n} \sigma_i = T\}$.
	 As was shown in \cite[Theorem 1]{Desio.Seiringer.2025}, we then have for $T>0$
	 	\begin{equation}
	 		\label{Equation: Dyson-type expansion of semigroup}
	 		\langle \Omega, e^{-TH(P)} \Omega \rangle = e^{-T|P|^2/2} + \sum_{n=1}^\infty \sum_{\pi \in \mathcal W_{2n}} \int_{\triangle_{2n, T}} \mathrm d \sigma \int_{\R^{dn}} \mathrm dk \prod_{i=1}^n |v(k_i)|^2 \operatorname{exp}\Big(- \sum_{j=0}^{2n} \sigma_j \mathcal E_P^{(\pi, j)}(k)\Big).
	 	\end{equation}
	 We will now bring \eqref{Equation: Dyson-type expansion of semigroup} into a point process representation which generalizes the one which we derived in \cite{Polzer.2023} for the Fröhlich model starting from the Feynman--Kac formula. On that note, let us denote for $T> 0$ by
	 
	 \begin{equation*}
	 	\mathcal I_T \coloneqq \{[s, t]: 0 \leq s < t \leq T\}, \quad \mathcal I \coloneqq \{[s, t]: 0 \leq s < t\}
	 \end{equation*}
	 the set of compact intervals contained in $[0, T]$ and $[0, \infty)$ respectively and by
	 \begin{equation*}
	 	\mathcal I_T^\cup \coloneqq \bigcup_{n=0}^\infty  \mathcal I_T^n, \quad 	\mathcal I^\cup \coloneqq \bigcup_{n=0}^\infty \mathcal I^n
	 \end{equation*}
	 the sets of finite collections of the latter (where $\mathcal I_T^0 = \mathcal I^0  = \{\emptyset\}$). For an interval $[s, t] \in \mathcal I$, we denote its length by $|[s, t]|\coloneqq t-s$. 
	 For $\xi \in \mathcal I^\cup$, we denote by $N(\xi)$ the number of intervals contained in $\xi$.
 	 If $\xi \neq \emptyset$, we define the symmetric matrix $\Sigma(\xi) \in \R^{dN(\xi) \times dN(\xi) } $ by 
	 \begin{equation*}
	 	\Sigma(\xi) \coloneqq 
	 	\begin{pmatrix}
	 		\Sigma_{11} & \hdots & \Sigma_{1N(\xi)} \\
	 		\vdots &  & \vdots \\ 
	 		\Sigma_{N(\xi)1} & \hdots & \Sigma_{N(\xi)N(\xi)} 
	 	\end{pmatrix}
	 	\quad \text{ where }\Sigma_{ij} \coloneqq |\xi_i \cap \xi_j| I_{d}
	 \end{equation*}
	 for all $i, j\in \{1, \hdots, N(\xi)\}$, where $I_d \in \R^{d\times d}$ is the identity matrix. Moreover, we define for $T\geq 0$
	 \begin{equation*}
	 	\Sigma_T(\xi) \coloneqq \Sigma\big(\xi_1, \hdots, \xi_{N(\xi)}, [0, T]\big).
	 \end{equation*}
	 We denote by $\Xi_T$ the law of a Poisson point process on $\mathcal I_T$ with intensity measure $\1_{\{0 < s < t < T\}} \mathrm ds \mathrm dt$, where we identify  $\mathcal I_T$ with $\big\{(s, t):\, 0\leq s < t \leq T\big\} \subset \R^2$.
	 Finally, we define
	 
	 \begin{equation*}
	 	F_T(\xi, P) \coloneqq \int_{\R^{d N(\xi)}} \operatorname{exp}\bigg(- \frac{1}{2} \bigg \langle \begin{pmatrix}
	 		k \\
	 		P
	 	\end{pmatrix}, \Sigma_T(\xi)
	 	\begin{pmatrix}
	 		k \\
	 		P
	 	\end{pmatrix} \bigg \rangle \bigg)  \prod_{i=1}^{N(\xi)} |v(k_i)|^2 e^{- |\xi_i| \omega(k_i)} \, \mathrm dk
	 \end{equation*}
	 if $\xi \neq \emptyset$ and $F_T(\emptyset, P) \coloneqq e^{-T|P|^2/2}$.

	 \begin{prop}
	 	\label{Proposition: Point process representation of semigroup}
	 	We have for all $T\geq 0$ and $P\in \R^d$
	 	 \begin{equation*}
	 		\langle \Omega, e^{-TH(P)} \Omega \rangle = e^{T^2/2 } \int_{\mathcal I_T^\cup } F_T(\xi, P) \, \Xi_T(\mathrm d\xi).
	 	\end{equation*}
	 \end{prop}
	 
	 \begin{proof}
	 	Let us fix some $\sigma \in \triangle_{2n, T}$ and $\pi \in \mathcal W_{2n}$. We define $r\in [0, T]^{2n}$ and $\xi \in \mathcal I_T^n$ by
	 	\begin{equation}
	 		\label{Equation: xi as a function of sigma and pi}
	 		\forall i \in \{1, \hdots, 2n\}:\, r_i = \sum_{j=0}^{i-1} \sigma_j, \quad 	\forall i \in \{1, \hdots, n\}:\, \xi_i \coloneqq [r_{\pi(2i-1)}, r_{\pi(2i)}].
	 	\end{equation}
	 	We have for all $k\in \R^{dn}$
	 	\begin{align*}
	 		&\sum_{j=0}^{2n} \sigma_j \mathcal E_P^{(\pi, j)}(k) = \sum_{j=0}^{2n} \sigma_j \Big(\tfrac{1}{2}|P|^2 - \sum_{l\in \mathcal M_j^\pi} P \cdot k_l + \tfrac{1}{2} \sum_{l, m \in \mathcal M_j^\pi} k_l \cdot k_m +  \sum_{l \in \mathcal M_j^\pi} \omega(k_l)  \Big) \\
	 		&= \tfrac{1}{2}|P|^2 T - \sum_{l=1}^n P \cdot k_l \sum_{j=0}^{2n} \sigma_j \1_{ \{l\in \mathcal M_j^\pi  \} } + \tfrac{1}{2}\sum_{l, m=1}^n k_l \cdot k_m \sum_{j=0}^{2n} \sigma_j \1_{ \{l, m\in \mathcal M_j^\pi  \} } + \sum_{l=1}^n \omega(k_l) \sum_{j=0}^{2n} \sigma_j \1_{ \{l\in \mathcal M_j^\pi  \} }.
	 	\end{align*}
	 	Now notice that we have for every $l \in \{1, \hdots, n\}$
	 	\begin{equation*}
	 		\sum_{j=0}^{2n} \sigma_j \1_{ \{l\in \mathcal M_j^\pi  \}} = \sum_{j= \pi(2l-1)}^{\pi(2l)-1} \sigma_j = |\xi_l|
	 	\end{equation*}
	 	and for all $l, m\in \{1, \hdots, n\}$ 
	 		\begin{equation*}
	 		\sum_{j=0}^{2n} \sigma_j \1_{ \{l, m\in \mathcal M_j^\pi  \}} = \sum_{j= \max\{\pi(2l-1), \pi(2m-1) \}  }^{\min \{ \pi(2l), \pi(2m)\} - 1} \sigma_j = |\xi_l \cap \xi_m|.
	 	\end{equation*}
	 	We hence obtain
	 	\begin{align}
	 		\label{Equation: rewrite E from sigma to xi}
	 		\sum_{j=0}^{2n} \sigma_j \mathcal E_P^{(\pi, j)}(-k) &= \tfrac{1}{2}|P|^2 T + \sum_{l=1}^n  |\xi_l| P \cdot k_l  + \tfrac{1}{2}\sum_{l, m=1}^n |\xi_l \cap \xi_m| k_l \cdot k_m + \sum_{l=1}^n |\xi_l| \omega(k_l) \nonumber \\
	 		&= \tfrac{1}{2} \Big \langle \begin{pmatrix}
	 			k \\
	 			P
	 		\end{pmatrix}, \Sigma_T(\xi)
	 		\begin{pmatrix}
	 			k \\
	 			P
	 		\end{pmatrix} \Big \rangle + \sum_{l=1}^n |\xi_l| \omega(k_l).
	 	\end{align}
	 	We define
	 	\begin{equation*}
	 		\mathcal J_{n, T} \coloneqq \big\{\big([s_1, t_1], \hdots, [s_n, t_n]\big) \in \mathcal I_{T}^n:\, 0<s_1 < \hdots < s_n< T, \quad \forall i \neq j:\, s_i \neq t_j,\, t_i \neq t_j \}.
	 	\end{equation*}
	 	Notice that
	 	\begin{equation*}
	 		\Phi : \mathcal W_{2n} \times \triangle_{2n, T} \to \mathcal J_{n, T},  \quad (\pi, \sigma) \mapsto \xi(\pi, \sigma)
	 	\end{equation*}
	 	with $\xi(\pi, \sigma)$ being defined by \eqref{Equation: xi as a function of sigma and pi} is a bijection with $\Phi(\pi, \cdot)$ having for every $\pi \in \mathcal W_{2n}$ functional determinant $\pm 1$. Combining \eqref{Equation: Dyson-type expansion of semigroup} and \eqref{Equation: rewrite E from sigma to xi}, we hence obtain\footnote{Recall that we identify $\mathcal I_T$ with $\big\{(s, t):\, 0\leq s < t \leq T\big\} \subset \R^2$ such that the integrals over $\mathcal J_{n, T}, \mathcal I_T^n$ are taken with respect to the Lebesgue measure.}
	 	\begin{align*}
	 		\langle \Omega, e^{-TH(P)} \Omega \rangle &= e^{-T|P|^2/2} + \sum_{n=1}^\infty \int_{ \mathcal J_{n, T} } \mathrm d \xi \int_{\R^{dn}} \mathrm dk \, \prod_{i=1}^n |v(k_i)|^2 e^{- |\xi_i| \omega(k_i)} \operatorname{exp}\Big(- \frac{1}{2} \Big \langle \begin{pmatrix}
	 			k \\
	 			P
	 		\end{pmatrix}, \Sigma_T(\xi)
	 		\begin{pmatrix}
	 			k \\
	 			P
	 		\end{pmatrix} \Big \rangle \Big)\\
	 		&= e^{-T|P|^2/2} + \sum_{n=1}^\infty \frac{1}{n!} \int_{\mathcal I_{T}^n} \mathrm d \xi \int_{\R^{dn}} \mathrm dk \, \prod_{i=1}^n |v(k_i)|^2 e^{- |\xi_i| \omega(k_i)} \operatorname{exp}\Big(- \frac{1}{2} \Big \langle \begin{pmatrix}
	 			k \\
	 			P
	 		\end{pmatrix}, \Sigma_T(\xi)
	 		\begin{pmatrix}
	 			k \\
	 			P
	 		\end{pmatrix} \Big \rangle \Big) \\
	 		&= e^{T^2/2 } \int_{\mathcal I_T^\cup } \Xi_T(\mathrm d\xi) \, F_T(\xi, P)
	 	\end{align*}
	 	where we used in the second to last inequality that $F_T(\xi, P)$ stays invariant under permutation of the intervals contained in $\xi$ and in the last equality that the intensity measure of $\Xi_T$ has total mass $T^2/2$.
	 \end{proof}
	 
	 We now notice that the function $F_{T}(\, \cdot\,, P)$ has a product structure under disjoint ``clusters'' of intervals. For all $\xi \in \mathcal I^\cup$ and $s<t$, we define
	 
	 \begin{equation*}
	 	\Sigma_{s, t}(\xi) \coloneqq \Sigma \big(\xi_1, \hdots, \xi_{N(\xi)}, [s, t]\big)
	 \end{equation*}
	 as well as
	 	 \begin{equation*}
	 	F_{s, t}(\xi, P) \coloneqq \int_{\R^{d N(\xi)}} \operatorname{exp}\bigg(- \frac{1}{2} \bigg \langle \begin{pmatrix}
	 		k \\
	 		P
	 	\end{pmatrix}, \Sigma_{s, t}(\xi)
	 	\begin{pmatrix}
	 		k \\
	 		P
	 	\end{pmatrix}\bigg \rangle \bigg)  \prod_{i=1}^{N(\xi)} |v(k_i)|^2 e^{- |\xi_i| \omega(k_i)} \, \mathrm dk.
	 \end{equation*}
	 Notice that if $\theta_s \xi \coloneqq (\xi_1 -s, \hdots, \xi_{N(\xi)} -s)$ denotes the translation of $\xi$ by $-s$ then
	 \begin{equation*}
	 	F_{s, t}(\xi, P) = F_{t-s}(\theta_s \xi, P).
	 \end{equation*}
	 If $\xi, \zeta \in \mathcal I_T^\cup$ are such that $\xi$ and $\zeta$ fall into $[0, s]$ and $[s, T]$ respectively, in the sense that $\xi_i \subseteq [0, s]$ for all $i \in \{1, \hdots, N(\xi)\}$ and $\zeta_j \subseteq [s, T]$ for all $j \in \{1, \hdots, N(\zeta)\}$, then for all $k \in \R^{dN(\xi)}$ and all $\tilde k \in \R^{dN(\zeta)}$
	 \begin{equation*}
	 	\bigg \langle \begin{pmatrix}
	 		k \\
	 		\tilde k \\
	 		P
	 	\end{pmatrix}, \Sigma_T(\xi, \zeta)
	 	\begin{pmatrix}
	 		k\\
	 		\tilde k \\
	 		P
	 	\end{pmatrix} \bigg \rangle =  \Big \langle \begin{pmatrix}
	 	k \\
	 	P
	 	\end{pmatrix}, \Sigma_s(\xi)
	 	\begin{pmatrix}
	 	k \\
	 	P
	 	\end{pmatrix} \Big \rangle
	 	+
	 	\Big \langle \begin{pmatrix}
	 		\tilde k \\
	 		P
	 	\end{pmatrix}, \Sigma_{s, T}(\zeta)
	 	\begin{pmatrix}
	 		\tilde k \\
	 		P
	 	\end{pmatrix} \Big \rangle
	 \end{equation*}
	 such that by Fubini's theorem
	 \begin{equation*}
	 	F_T((\xi, \zeta), P) = F_s(\xi, P)\cdot F_{s, T}(\zeta, P).
	 \end{equation*}
	 We will now apply this product structure in order to determine the renewal transform of the spectral measure of $H(P)$ with respect to $\Omega$. For $t\geq 0$, we define $\hat Y_t: \mathcal I^\cup \to \{0, 1\}$ by
	 \begin{equation*}
		 \hat Y_t(\xi) \coloneqq \begin{cases}
		 	1 \quad &\exists i\in \{1, \hdots, N(\xi)\} \text{ s.t. } t \in \xi_i \\
		 	0 \text &\text{else}.
		 \end{cases}
	 \end{equation*}

	 \begin{prop}
	 	\label{Proposition: Renewal transform of spectral measures}
	 	 The finite volume renewal transform $\mathbf P_{P, T}$ of the spectral measure of $H(P)$ with respect to $\Omega$ coincides with the distribution of $(\hat Y_t)_{0\leq t \leq T}$ under
	 	\begin{equation*}
	 		\widehat \Xi_{P, T}(\mathrm d\xi) \coloneqq \frac{e^{T^2/2}}{\big\langle \Omega, e^{-TH(P)} \Omega \big \rangle } F_T(\xi, P) \, \Xi_T(\mathrm d\xi).
	 	\end{equation*}
	 \end{prop}
	 
	 \begin{proof}
	 	As before, we identify $\mathcal I$ with $\big\{(s, t):\, 0\leq s < t\big\} \subset \R^2$.
	 	Let $\eta$ denote a Poisson point process on $\R^2$ with intensity measure $\1_{\{0 < s < t\}} \mathrm ds \mathrm dt$. For $T>0$, let $\eta_T$ denote the restriction of $\eta$ to $[0, T]^2$ such that $\eta_T \sim \Xi_T$. Notice that in particular
	 	\begin{equation*}
	 		\mathbb P(N(\eta_T) = 0) = e^{-T^2/2}.
	 	\end{equation*}
	 	For $0\leq s < t \leq T$, we denote by $\eta_{s, t}$ the restriction of $\eta_T$ to $(s, t] \times (s, t]$ and by $\tilde \eta_{s, t}$ the restriction of $\eta_T$ to $(s, t] \times (t, T]$.
	 	Let $0 = r_0 < r_1 < \hdots r_{k} < r_{k+1} = T$.
	 	Since $\eta_{0, r_1}, \eta_{r_1, r_2}, \hdots, \eta_{r_k, T}$ are independent with $\theta_{r_i}  \eta_{r_i, r_{i+1}} \overset{d}{=} \eta_{r_{i+1}-r_{i}}$ for all $i\in \{0, \hdots, k\}$, we obtain with Proposition \ref{Proposition: Point process representation of semigroup}
	 	\begin{equation*}
	 		\prod_{i=0}^{k} \big \langle\Omega, e^{-(r_{i+1}-r_i)H(P)} \Omega \big\rangle
	 		= \frac{1}{ \prod_{i=0}^{k} \mathbb P(N(\eta_{r_i, r_{i+1}}) = 0)} \mathbb E\Big[\prod_{i=0}^{k} F_{r_i, r_{i+1}}(\eta_{r_i, r_{i+1}}, P)\Big].
	 	\end{equation*}
	 	The event 
	 	\begin{equation*}
	 		A \coloneqq \{\hat Y_{r_1}(\eta_T) =  \hdots = \hat Y_{r_k}(\eta_T) = 0 \} = \{N(\tilde \eta_{0, r_1}) = N(\tilde \eta_{r_1, r_2}) = \hdots = N(\tilde \eta_{r_{k-1}, r_k}) = 0\},
	 	\end{equation*}
	 	is independent of $\big(\eta_{0, r_1}, \eta_{r_1, r_2}, \hdots, \eta_{r_k, T}\big)$ and conditionally on $A$ we have
	 	\begin{equation*}
	 		F_T(\eta_T, P) = \prod_{i=0}^{k} F_{r_i, r_{i+1}}(\eta_{r_i, r_{i+1}}, P).
	 	\end{equation*}
	 	We hence obtain
	 	\begin{equation*}
	 		\prod_{i=0}^{k} \big\langle \Omega, e^{-(r_{i+1}-r_i)H(P)} \Omega \big\rangle = \frac{1}{ \prod_{i=0}^{k} \mathbb P(N(\eta_{r_i, r_{i+1}}) = 0)} \mathbb E\big[F_{T}(\eta_{T}, P)|A\big].
	 	\end{equation*}
	 	Using that
	 	\begin{align*}
	 		\mathbb P(N(\eta_T) = 0) &= \mathbb P\big(\{N(\eta_{0, r_1}) = \hdots = N(\eta_{r_k, T}) = 0\} \cap A\big) = \mathbb P(A) \cdot \prod_{i=0}^k \mathbb P(N(\eta_{r_i, r_{i+1}})= 0)
	 	\end{align*}
	 	we hence obtain with Proposition \ref{Proposition: Point process representation of semigroup}
	 	\begin{equation*}
	 		\frac{\prod_{i=0}^{k} \big\langle \Omega, e^{-(r_{i+1} -r_i)H(P)} \Omega \big\rangle}{\big\langle \Omega, e^{-T H(P)} \Omega \big\rangle}
	 		= \frac{\mathbb E\big[F_{T}(\eta_{T}, P) \1_A \big] }{\mathbb P(A) \prod_{i=0}^{k} \mathbb P(N(\eta_{r_i, r_{i+1}}) = 0)} \cdot \frac{\mathbb P(N(\eta_T) = 0)}{\mathbb E[F_{T}(\eta_{T}, P)]}
	 		=\frac{\mathbb E\big[F_{T}(\eta_{T}, P) \1_A \big] }{\mathbb E\big[F_{T}(\eta_{T}, P)\big] }
	 	\end{equation*}
	 	such that we obtain with \eqref{Equation: fractions of Z in terms of renewal transform} and the definition of $\widehat \Xi_{P, T}$
	 	\begin{equation*}
	 		\mathbf P_{P, T}(Y_{r_1} = \hdots = Y_{r_k} = 0)= \widehat \Xi_{P, T}( \hat Y_{r_1} = \hdots = \hat Y_{r_k} = 0).
	 	\end{equation*}
	 	Hence, the probability measure $\mathbf P_{P,T}$ coincides on an intersection stable generator with the distribution of $(\hat Y_t)_{0\leq t \leq T}$ under $\widehat \Xi_{P, T}$ and the claim follows.
	 \end{proof}
	 
	 \begin{cor}
	 	\label{Corollary: Point process representation of vacuum overlap}
	 	For all $P \in \R^d$
	 	\begin{equation*}
	 		\rho(P) = \lim_{T\to \infty} \mathbb E_{\widehat{\Xi}_{P, T}}[\hat D_T/T] \quad \text{ where }\quad \hat D_T \coloneqq \int_0^T (1-\hat Y_s) \, \mathrm ds.
	 	\end{equation*}
	 \end{cor}
	 
	 \begin{proof}
	 	This follows directly from Proposition \ref{Proposition: Renewal transform of spectral measures} and \cite[Theorem 4.5]{HinrichsPolzer.2025}.
	 \end{proof}
	 
	 Our assumption that $|k|^2 \mapsto |v(k)|^2$ is completely monotone and that $|k|^2 \mapsto  \omega(k)$ is a Bernstein function allows us to rewrite $F_T$ in a more convenient form. Let $(X_t)_{t \geq 0}$ be a Brownian motion on $\R^d$. For an interval $[s, t] \in \mathcal I$ we set
	 \begin{equation*}
	 	X_{[s, t]} \coloneqq X_{s, t} \coloneqq X_t - X_s.
	 \end{equation*}
	 Notice that the matrix $\Sigma(\xi)$ is for $\xi \in \mathcal I^\cup$ the covariance matrix of the centered Gaussian vector 
	 \begin{equation*}
	 	X^\xi \coloneqq \big(X_{\xi_i}\big)_{1 \leq i \leq N(\xi)}.
	 \end{equation*}
	 For $u \in [0, \infty)^{N(\xi)}$, we define the Gaussian measure
	 \begin{equation}
	 	\label{Equation: Definition of P(xi, u)}
	 	\mathbf P_{\xi, u}(\mathrm dX)\coloneqq \frac{1}{\phi(\xi, u)}\operatorname{exp}\Big(-\frac{1}{2}\sum_{i=1}^{N(\xi)} u_i^2 |X_{\xi_i}|^2 \Big) \, \mathcal W(\mathrm d X)
	 \end{equation}
	 where $\mathcal W$ is as before the distribution of Brownian motion and where $\phi(\xi, u)$ acts as a normalization constant. We denote the expected value taken with respect to $\mathbf P_{\xi, u}$ by $\mathbf E_{\xi, u}$ and define
	 \begin{equation*}
	 	\sigma^2_T(\xi, u) \coloneqq \frac{1}{d} \mathbf E_{\xi, u}[|X_{0, T}|^2].
	 \end{equation*}
	 We point out that we have the representation \cite[Equation (3.1)]{BetzPolzer.2022},
	\begin{equation}
		\label{Equation: representation of sigma as L2 distance}
		\sigma_T^2(\xi, u) = \operatorname{dist}_{L^2}\big(B_{0, T}, \operatorname{span}\{u_i B_{s_i, t_i} + Z_i:\, 1\leq i \leq N(\xi)\}\big)^2
	\end{equation}
	where $B = (B_t)_{t\geq 0} $ is an one dimensional Brownian motion and $(Z_i)_{i\in \N}$ is an iid family of $\mathcal N(0, 1)$ distributed random variables which is independent of $B$. In particular, we have
	\begin{equation}
		\label{Equation: Trivial estiamtes on sigma}
		\hat D_T(\xi) \leq \sigma_T^2(\xi, u) \leq T.
	\end{equation}
	\begin{prop}
		\label{Proposition: Useful form of F}
		We have for $\xi \neq \emptyset$
	\begin{equation*}
				F_T(\xi, P) = \int_{(0, \infty)^{N(\xi)}} \bigotimes_{i=1}^{N(\xi)} \kappa(|\xi_i|, \mathrm d u_i) \, \phi(\xi, u) e^{-|P|^2 \sigma_T^2(\xi, u)/2}.
	\end{equation*}
	\end{prop}
	
	\begin{proof}
		Since $\Xi_T(\xi)$ is the covariance matrix of $(X^\xi, X_{0, T})$ we have
		\begin{equation*}
			\operatorname{exp}\Big(- \frac{1}{2} \Big \langle \begin{pmatrix}
				k \\
				P
			\end{pmatrix}, \Sigma_T(\xi)
			\begin{pmatrix}
				k \\
				P
			\end{pmatrix}\Big \rangle \Big)
			= \mathbb E_\mathcal W\Big[ e^{\mathrm i P \cdot X_{0, T}} \prod_{i=1}^{N(\xi)} e^{\mathrm i k_i \cdot X_{\xi_i}} \Big].
		\end{equation*}
		First, let us assume that $v \in L^2(\R^d)$. Then we have by Fubini's theorem
		\begin{equation*}
			F_T(\xi, P) = \mathbb E_\mathcal W\Big[ e^{\mathrm i P \cdot X_{0, T}} \prod_{i=1}^{N(\xi)} \int_{\R^d} e^{\mathrm i k_i \cdot X_{\xi_i}} |v(k_i)|^2 e^{-\omega(k_i)|\xi_i|} \mathrm dk_i \Big].
		\end{equation*}
		Using the definitions of $\kappa, \tilde \kappa$, we have for all $t>0$ and $x\in \R^d$
		\begin{align*}
			\int_{\R^d} e^{\mathrm i k \cdot x} |v(k)|^2 e^{-\omega(k) t} \, \mathrm dk &= \int_{(0, \infty)} \tilde \kappa(t, \mathrm du) \int_{\R^d} \mathrm dk \, e^{\mathrm i k \cdot x} e^{-u |k|^2/2 } \\
			&= \int_{(0, \infty)} \tilde \kappa(t, \mathrm du) \big(2\pi/u\big)^{d/2} e^{-|x|^2/(2u) } 
			= \int_{(0, \infty)} \kappa(t, \mathrm du) \, e^{-u^2|x|^2/2 } 
		\end{align*}
		such that
		\begin{align}
			F_T(\xi, P) &= \mathbb E_{\mathcal W}\Big[ e^{\mathrm i P \cdot X_{0, T}} \prod_{i=1}^{N(\xi)} \int_{(0, \infty)} \kappa(|\xi_i|, \mathrm du_i) \, e^{-u_i^2|X_{\xi_i}|^2/2 } \label{Equation: F as expected value over product of w} \Big] \\
			&=\int_{(0, \infty)^{N(\xi)}}  \bigotimes_{i=1}^{N(\xi)} \kappa(|\xi_i|, \mathrm du_i) \, \mathbb E_\mathcal W\Big[e^{\mathrm i P \cdot X_{0, T}}\prod_{i=1}^{N(\xi)}  e^{-u_i^2|X_{\xi_i}|^2/2 }   \Big]  \nonumber \\
			&= \int_{(0, \infty)^{N(\xi)}} \bigotimes_{i=1}^{N(\xi)} \kappa(|\xi_i|, \mathrm d u_i) \, \phi(\xi, u) e^{-|P|^2 \sigma_T^2(\xi, u)/2} 	\label{Equation: Nice form of F for sqaure integrable v}
		\end{align}
		where the last equality follows from the fact that $\mathbf P_{\xi, u}$ is a rotationally symmetric centered Gaussian measure such that $X_{0, T}$ has distribution $\mathcal N(0, \sigma^2_T(\xi, u)I_d)$ under $\mathbf P_{\xi, u}$. Hence, the statement of Proposition \ref{Proposition: Useful form of F} holds under the additional assumption that $v\in L^2(\R^d)$. Now, let us consider the general case in which $v$ might not be square integrable. By Assumption \ref{Assumption: Assumption 1}, there exists measures $\mu, \nu_t$ such that
		\begin{equation*}
			|v(k)|^2 = \int_{(0, \infty)} \mu(\mathrm du) \, e^{-u |k|^2/2}, \quad e^{-\omega(k)t} = \int_{(0, \infty)} \nu_t(\mathrm du) \, e^{- u |k|^2/2}
		\end{equation*}
		for all $k\in \R^d$ and $t\geq 0$. The kernel $\tilde \kappa$ in \eqref{Equation: Definition of kappatilde} is then given by the convolution
		\begin{equation*}
			\tilde \kappa(t, \, \cdot \,) = \mu* \nu_t.
		\end{equation*}
		We define the function $v_\varepsilon$ by
		\begin{equation}
			\label{Equation: Definition regularized v}
			v_\varepsilon(k) \coloneqq \Big(\int_{(0, \infty)} \mu_\varepsilon(\mathrm du) \, e^{-u |k|^2/2}\Big)^{1/2} \quad \text{ where }  \mu_\varepsilon(\mathrm du) \coloneqq \1_{\{\epsilon < u < 1/\varepsilon \}} \, \mu(\mathrm du)
		\end{equation}
		and define $\tilde \kappa_\varepsilon(t, \, \cdot \,) \coloneqq \mu_\varepsilon * \nu_t$ and $\kappa_\varepsilon(t, \, \cdot \,)$ accordingly such that
		\begin{equation*}
			\int_{(0, \infty)} \kappa_\varepsilon(t, \mathrm du) \, f(u) = \int_{(0, \infty)} \tilde \kappa_\varepsilon(t, \mathrm du) \, (2\pi/u)^{d/2} f(u^{-1/2})
		\end{equation*}
		for all non-negative measurable $f$. 
		Then $v_\varepsilon \in L^2(\R^d)$ for all $\varepsilon>0$ and $|v_\varepsilon|^2 \uparrow |v|^2$ as $\varepsilon\downarrow 0$. By the monotone convergence theorem and \eqref{Equation: Nice form of F for sqaure integrable v}
		\begin{align*}
				F_T(\xi, P) &\coloneqq \lim_{\varepsilon \downarrow 0} \int_{\R^{d N(\xi)}} \operatorname{exp}\Big(- \frac{1}{2} \Big \langle \begin{pmatrix}
					k \\
					P
				\end{pmatrix}, \Sigma_T(\xi)
				\begin{pmatrix}
					k \\
					P
				\end{pmatrix} \Big \rangle \Big)  \prod_{i=1}^{N(\xi)} |v_\varepsilon(k_i)|^2 e^{- |\xi_i| \omega(k_i)} \, \mathrm dk \\
				&= \lim_{\varepsilon \downarrow 0} \int_{(0, \infty)^{N(\xi)}} \bigotimes_{i=1}^{N(\xi)} \kappa_\varepsilon(|\xi_i|, \mathrm d u_i) \, \phi(\xi, u) e^{-|P|^2 \sigma_T^2(\xi, u)/2} \\
				&= \int_{(0, \infty)^{N(\xi)}} \bigotimes_{i=1}^{N(\xi)} \kappa(|\xi_i|, \mathrm d u_i) \, \phi(\xi, u) e^{-|P|^2 \sigma_T^2(\xi, u)/2} . \qedhere
		\end{align*}
	\end{proof}

	Considering the representation of $F_P$ given in Proposition \ref{Proposition: Useful form of F}, it will be convenient to interpret $\widehat{\Xi}_{P, T}$ as the marginal distribution of suitable measure on $\mathcal Y_T \coloneqq \bigcup_{n=0}^\infty \mathcal I_T^n \times (0, \infty)^n$. On that note, we define the measure $\widehat{\Theta}_{P, T}$ on $\mathcal Y_T$ by
	\begin{equation*}
		\widehat{\Theta}_{P, T}(\mathrm d \xi \mathrm du) \coloneqq \frac{e^{T^2/2}}{\langle \Omega, e^{-TH(P)} \Omega \rangle } \Xi_T(\mathrm d\xi) \bigotimes_{i=1}^{N(\xi)} \kappa(|\xi_i|, \mathrm d u_i) \, \phi(\xi, u) e^{-|P|^2 \sigma^2(\xi, u)/2}
	\end{equation*}
	where Propositions \ref{Proposition: Point process representation of semigroup} and \ref{Proposition: Useful form of F} guarantee that $\widehat{\Theta}_{P, T}$ is a probability measure. Notice that Proposition \ref{Proposition: Useful form of F} implies that we indeed recover $\widehat{\Xi}_{P, T}$ as the marginal of $\widehat{\Theta}_{P, T}$ after integrating out $u$.
	\begin{cor}
		\label{Corollary: Point process representation of fraction of partition functions}
		For all $P, \tilde P \in \R^d$
		\begin{equation}
			\label{Equation: Point process representation of fraction of partition functions}
			\frac{\langle \Omega, e^{-TH(\tilde P)} \Omega \rangle}{\langle \Omega, e^{-TH(P)} \Omega \rangle} = 	\mathbb E_{\widehat \Theta_{P, T}}[e^{-(|\tilde P|^2 - |P|^2) \sigma_T^2/2}].
		\end{equation}
	\end{cor}
	
	\begin{proof}	
		By Propositions \ref{Proposition: Point process representation of semigroup} and \ref{Proposition: Useful form of F} we have 
		\begin{align*}
			\frac{\langle \Omega, e^{-TH(\tilde P)} \Omega \rangle}{\langle \Omega, e^{-TH(P)} \Omega \rangle} &= \frac{\int_{\mathcal I_T^\cup} \Xi_T(\mathrm d\xi)\int_{(0, \infty)^{N(\xi)}} \bigotimes_{i=1}^{N(\xi)} \kappa(|\xi_i|, \mathrm d u_i) \, \phi(\xi, u) e^{-(|\tilde P|^2 - |P|^2) \sigma^2(\xi, u)/2} e^{- |P|^2 \sigma^2(\xi, u)/2} }{\int_{\mathcal I_T^\cup} \Xi_T(\mathrm d\xi)\int_{(0, \infty)^{N(\xi)}} \bigotimes_{i=1}^{N(\xi)} \kappa(|\xi_i|, \mathrm d u_i) \, \phi(\xi, u) e^{- |P|^2 \sigma^2(\xi, u)/2}} \\
			&=\mathbb E_{\hat \Theta_{P, T}}[e^{-(|\tilde P|^2 - |P|^2) \sigma_T^2/2}]. \qedhere
		\end{align*}
		
	\end{proof}

	\section{Proof of Theorem \ref{Theorem: Lower bound on EMR}}
	
	\begin{proof}[Proof of Theorem \ref{Theorem: Lower bound on EMR}]
		We fix some $P_0 \in \R^d$ and define for $P \in \R^d$ with $|P| \geq |P_0|$
		\begin{equation*}
			A(P) \coloneqq H(P_0) + \delta(P, P_0) \langle \Omega, \, \cdot \, \rangle \Omega.
		\end{equation*}
		By \eqref{Equation: Trivial estiamtes on sigma}, Corollary \ref{Corollary: Point process representation of fraction of partition functions} and Proposition \ref{Proposition: Renewal transform under rank-one perturbations}
		\begin{equation*}
			\frac{\langle \Omega, e^{-TH(P)} \Omega \rangle}{\langle \Omega, e^{-TH(P_0)} \Omega \rangle} = 	\mathbb E_{\widehat \Theta_{P_0, T}}\big[e^{-\delta(P, P_0)  \sigma_T^2}\big] \leq \mathbb E_{\widehat\Xi_{P_0, T}}\big[e^{-\delta(P, P_0)  \hat D_T}\big] = \frac{\langle \Omega, e^{-TA(P)} \Omega \rangle}{\langle \Omega, e^{-TH(P_0)} \Omega \rangle}
		\end{equation*}
		such that
		\begin{equation*}
			E(P) = -\lim_{T\to \infty} \frac{1}{T} \log \, \langle \Omega, e^{-TH(P)} \Omega \rangle \geq -\lim_{T\to \infty} \frac{1}{T} \log \, \langle \Omega, e^{-TA(P)} \Omega \rangle = \inf \operatorname{supp} \mu_{P} \eqqcolon f(P)
		\end{equation*}
		where $\mu_{P}$ denotes the spectral measure of $A(P)$ with respect to $\Omega$. From the Sherman--Morrison formula (or alternatively from \eqref{Equation: Sherman–Morrison formula}) one obtains
		\begin{equation*}
			\tilde s(z) \coloneqq \big\langle \Omega, (A(P) - z)^{-1} \Omega \big\rangle  = \frac{s(z)}{1 + \delta(P, P_0)s(z)} \quad \text{ where }s(z) \coloneqq \langle \Omega, (H(P_0) - z)^{-1} \Omega \rangle
		\end{equation*}
		for all $z\in \C \setminus \R$. Notice that $\lambda \mapsto s(\lambda)$ is strictly increasing on the interval $(E(P_0), E(P_0) + \Delta(P_0))$.
		Since finite rank perturbations leave the essential spectrum invariant, all mass of $\mu_{P}$ in $[E(P_0), E_{\operatorname{ess}}(P_0))$ comes from its pure point part. For $f(P) \geq E(P_0) + \Delta(P_0)$ there is nothing to show. So let us assume that $f(P) < E(P_0) + \Delta(P_0)$. Then $f(P)$ is an atom of $\mu_{P}$ and hence the unique solution $\lambda \in \big(E(P_0), E(P_0) + \Delta(P_0)\big)$ to the equation $ \delta(P, P_0)s(\lambda) = -1$.
		Since for all $\lambda \in \big(E(P_0), E(P_0) + \Delta(P_0)\big)$
		\begin{equation*}
			s(\lambda)  = \frac{\rho(P_0)}{E(P_0) - \lambda} + \int_{[E(P_0) + \Delta(P_0), \infty)} \frac{\mathrm \mu_{P_0}(\mathrm dx)}{x - \lambda} \leq \frac{\rho(P_0)}{E(P_0) - \lambda} +  \frac{1-\rho(P_0)}{E(P_0) + \Delta(P_0) - \lambda} ,
		\end{equation*}
		we find that $f(P)$ is no smaller than the unique solution $\lambda \in \big(E(P_0), E(P_0) + \Delta(P_0)\big)$ to the equation
		\begin{equation}
		\label{Equation: Eq for lower bound on E}
			\delta(P, P_0)\Big( \frac{\rho(P_0)}{E(P_0) - \lambda} +  \frac{1-\rho(P_0)}{E(P_0) + \Delta(P_0) - \lambda} \Big) = -1.
		\end{equation}
		Solving \eqref{Equation: Eq for lower bound on E} for $\lambda - E(P_0)$ yields the claim. 
	\end{proof}

	\section{Proof of Theorem \ref{Theorem: Feyman--Kac formula}}
	
	\begin{lemma}
		\label{Lemma: F as expected value of product of w}
		For $\Xi_T$-almost all $\xi$ we have for all $P \in \R^d$
		\begin{equation*}
			F_T(\xi, P) = \mathbb E_\mathcal W\Big[e^{\mathrm i P\cdot X_{0, T}}\prod_{i=1}^{N(\xi)} w(|\xi_i|, X_{\xi_i} )   \Big].
		\end{equation*}
	\end{lemma}
	
	\begin{proof}
		Under the additional assumption $v\in L^2(\R^d)$ we have already shown the statement in \eqref{Equation: F as expected value over product of w}. For general $v$, we approximate $v$ as in the proof of Proposition \ref{Proposition: Useful form of F} by $v_\varepsilon$ and define $w_\varepsilon$ and $F_{T, \varepsilon}$ accordingly. We obtain with the monotone convergence theorem
		\begin{equation*}
			F_T(\xi, P) =  \lim_{\varepsilon\downarrow 0} F_{T, \varepsilon}(\xi, P) = \lim_{\varepsilon\downarrow 0} \mathbb E_\mathcal W\Big[e^{\mathrm i P\cdot X_{0, T}}\prod_{i=1}^{N(\xi)} w_\varepsilon(|\xi_i|, X_{\xi_i} )   \Big].
		\end{equation*}
		For $P=0$, another application of the monotone convergence theorem yields
		\begin{equation*}
			F_T(\xi, 0) = \mathbb E_\mathcal W\Big[\prod_{i=1}^{N(\xi)} w(|\xi_i|, X_{\xi_i} )   \Big].
		\end{equation*}
		By Proposition \ref{Proposition: Point process representation of semigroup}, we have $F_T(\xi, 0)< \infty$ for $\Xi_T$ almost all $\xi$, such that we obtain for those $\xi$
		\begin{equation*}
			F_T(\xi, P) = \lim_{\varepsilon\downarrow 0} \mathbb E_\mathcal W\Big[e^{\mathrm i P\cdot X_{0, T}}\prod_{i=1}^{N(\xi)} w_\varepsilon(|\xi_i|, X_{\xi_i} )   \Big] = \mathbb E_\mathcal W\Big[e^{\mathrm i P\cdot X_{0, T}}\prod_{i=1}^{N(\xi)} w(|\xi_i|, X_{\xi_i} )   \Big]
		\end{equation*}
		for all $P\in \R^d$ by the dominated convergence theorem.
	\end{proof}

	\begin{proof}[Proof of Theorem \ref{Theorem: Feyman--Kac formula}]
		We have by Proposition \ref{Proposition: Point process representation of semigroup} and Lemma \ref{Lemma: F as expected value of product of w}
		\begin{align*}
			\big \langle \Omega, e^{-TH(P)} \Omega \big \rangle &= e^{T^2/2 } \int_{\mathcal I_T^\cup } \Xi_T(\mathrm d\xi) \, F_T(\xi, P) \\
			&= \sum_{n=0}^\infty \frac{1}{n!} \int_{\mathcal I_T^n} \mathrm d \xi \, \mathbb E_\mathcal W\Big[e^{\mathrm i P\cdot X_{0, T}}\prod_{i=1}^{N(\xi)} w(|\xi_i|, X_{\xi_i} )   \Big] \\
			&= \mathbb E_{\mathcal W}\bigg[ e^{\mathrm i P\cdot X_{0, T}} \sum_{n=0}^\infty \frac{1}{n!} \Big(\frac{1}{2}\int_{[0, T]^2} w(|t-s|, X_{s, t}) \, \mathrm ds \mathrm dt \Big)^n \bigg] \\
			&= \mathbb E_{\mathcal W}\bigg[e^{\mathrm i P\cdot X_{0, T}} \exp\Big(\frac{1}{2}\int_{[0, T]^2} w(|t-s|, X_{s, t}) \, \mathrm ds \mathrm dt   \Big)  \bigg]
		\end{align*}
		where we might interchange the order of integration for $P=0$ by non-negativity of the integrand and then for general $P$ by the dominated convergence theorem.
	\end{proof}
	
	\section{Proof of Theorem \ref{Theorem: Convergence of means square displacement}}

	\begin{prop}
		\label{Proposition: Point process representation of mean square displacement}
		We have for all $T>0$ and $P\in \R^d$
		\begin{equation*}
			\hat \sigma_T^2(P)
			 =  \frac{1}{T}\mathbb E_{\widehat{\Theta}_{P, T}}[\sigma^2_T].
		\end{equation*}
	\end{prop}
	
	\begin{proof}
		We start by noticing that
		\begin{align*}
			e^{-|P|^2 \sigma_T^2(\xi, u)/2} \sigma_T^2(\xi, u)  &= -\frac{1}{|P|} \partial_{|P|} \mathbf E_{\xi, u} \big[ e^{\mathrm i P\cdot X_{0, T}} \big] \\
			&= -\frac{\mathrm i}{|P|}\mathbf E_{\xi, u} \Big[ \tfrac{P \cdot X_{0, T}}{|P|} e^{\mathrm i P\cdot X_{0, T}} \Big] = \frac{1}{|P|^2} \mathbf E_{\xi, u} \Big[ (P \cdot X_{0, T})^2 \operatorname{sinc}(P\cdot X_{0, T}) \Big]
		\end{align*}
		for all $P\in \R^d \setminus \{0\}$ such that
		\begin{equation}
			\label{Equation: integrand in terms of Gaussian}
			e^{-|P|^2 \sigma_T^2(\xi, u)/2} \sigma_T^2(\xi, u) = \mathbf E_{\xi, u}\big[f(P, X_{0, T}) \big]
		\end{equation}
		for all $P\in \R^d$ where
		\begin{align*}
			f(P, x) \coloneqq \begin{cases}
				\frac{(P \cdot x)^2}{|P|^2} \operatorname{sinc}(P\cdot x) \quad &\text{ for }P\neq 0, \, x\in \R^d \\
				(v \cdot x)^2 \quad &\text{ for }P= 0,\, x\in \R^d		
			\end{cases}
		\end{align*}
		for some arbitrary unit vector $v \in \R^d$.
		We now proceed similarly to the proof of Theorem \ref{Theorem: Feyman--Kac formula}: we have
		\begin{align*}
			\mathbb E_{\widehat{\Theta}_{P, T}}[\sigma^2_T] &= \frac{e^{T^2/2}}{\langle \Omega, e^{-TH(P)} \Omega \rangle}\int_{\mathcal I_T^\cup} \Xi_T(\mathrm d \xi) \, \int_{(0, \infty)^{N(\xi)}} \bigotimes_{i=1}^{N(\xi)} \kappa(|\xi_i|, \mathrm du_i) \, \phi(\xi, u) e^{-|P|^2 \sigma_T^2(\xi, u)/2} \sigma_T^2(\xi, u) \\
			&= \frac{e^{T^2/2}}{\langle \Omega, e^{-TH(P)} \Omega \rangle} \int_{\mathcal I_T^\cup} \Xi_T(\mathrm d \xi) \, \int_{(0, \infty)^{N(\xi)}} \bigotimes_{i=1}^{N(\xi)} \kappa(|\xi_i|, \mathrm du_i) \, \phi(\xi, u) \mathbf E_{\xi, u} \big[ f(P, X_{0, T})\big] \\
			&= \frac{e^{T^2/2}}{\langle \Omega, e^{-TH(P)} \Omega \rangle} \int_{\mathcal I_T^\cup} \Xi_T(\mathrm d \xi) \, \int_{(0, \infty)^{N(\xi)}} \bigotimes_{i=1}^{N(\xi)} \kappa(|\xi_i|, \mathrm du_i)  \mathbb E_{\mathcal W} \Big[ f(P, X_{0, T}) \prod_{i=1}^{N(\xi)} e^{-u_i^2 |X_{\xi_i}|^2/2  }   \Big] \\
			&= \frac{e^{T^2/2}}{\langle \Omega, e^{-TH(P)} \Omega \rangle} \int_{\mathcal I_T^\cup} \Xi_T(\mathrm d \xi) \,   \mathbb E_{\mathcal W} \Big[f(P, X_{0, T}) \prod_{i=1}^{N(\xi)} w(|\xi_i|, X_{\xi_i})   \Big] \\
			&= \frac{1}{\langle \Omega, e^{-TH(P)} \Omega \rangle} \mathbb E_\mathcal W\bigg[f(P, X_{0, T})\exp\Big(\frac{1}{2}\int_{[0, T]^2} w(|t-s|, X_{s, t}) \, \mathrm ds \mathrm dt   \Big)  \bigg]
		\end{align*}
		where we might exchange the order of integration for $P=0$ by non-negativity of the integrand and then for $P \neq 0$ by the dominated convergence theorem. The statement follows from the representation of $\langle \Omega, e^{-TH(P)} \Omega \rangle$ given by Theorem \ref{Theorem: Feyman--Kac formula}.
	\end{proof}

	\begin{proof}[Proof of Theorem \ref{Theorem: Convergence of means square displacement}]
		We write with \eqref{Equation: E from semigroup} and Propositions \ref{Proposition: Point process representation of semigroup} and \ref{Proposition: Useful form of F}
		\begin{equation*}
			E(P) = \lim_{T\to \infty} - \frac{1}{T} \log \langle \Omega, e^{-TH(P)} \Omega \rangle = \lim_{T\to \infty} \phi_T(P^2)
		\end{equation*}
		where
		\begin{equation*}
			\phi_T(\alpha) \coloneqq  - \frac{1}{T}\log\Big( e^{T^2/2 } \int_{\mathcal I_T^\cup } \Xi_T(\mathrm d\xi) \, \int_{(0, \infty)^{N(\xi)}} \bigotimes_{i=1}^{N(\xi)} \kappa(|\xi_i|, \mathrm d u_i) \, \phi(\xi, u) e^{-\alpha \sigma_T^2(\xi, u)/2}  \Big).
		\end{equation*}
		By convexity of cumulant generating functions, $\phi_T$ is concave for every $T>0$. Hence, $\mathcal E$ is, as the pointwise limit of concave functions, concave and we have with Proposition \ref{Proposition: Point process representation of mean square displacement}
		\begin{equation*}
			(\partial^+ \mathcal E)(|P|^2) \leq \liminf_{T\to \infty} 	\phi_T'(|P|^2) = \tfrac{1}{2} \liminf_{T\to \infty} \mathbb E_{\widehat{\Theta}_{P, T}}[\sigma^2_T/T] = \tfrac{1}{2}\liminf_{T\to \infty}	\hat \sigma_T^2(P)
		\end{equation*}
		for all $P\in \R^d$ as well as
		\begin{equation*}
			(\partial^-  \mathcal E)(|P|^2) \geq \limsup_{T\to \infty} 	\phi_T'(|P|^2) = \tfrac{1}{2} \limsup_{T\to \infty} \mathbb E_{\widehat{\Theta}_{P, T}}[\sigma^2_T/T] = \tfrac{1}{2} \limsup_{T\to \infty}	\hat \sigma_T^2(P)
		\end{equation*}
		for all $P\in \R^d \setminus \{0\}$.
		Using Proposition \ref{Proposition: Point process representation of mean square displacement}, Corollary \ref{Corollary: Point process representation of vacuum overlap} and \eqref{Equation: Trivial estiamtes on sigma}, we obtain
		\begin{equation*}
			\liminf_{T\to \infty} \hat \sigma^2_T(P) = \liminf_{T\to \infty}  \mathbb E_{\widehat{\Theta}_{P, T}}[\sigma^2_T/T] \geq \lim_{T\to \infty} \mathbb E_{\widehat{\Xi}_{P,T}} [\hat D_T/T] = \rho(P).
		\end{equation*}
		as well as 
		\begin{equation*}
			\limsup_{T\to \infty} \hat \sigma^2_T(P) = \limsup_{T\to \infty}  \mathbb E_{\widehat{\Theta}_{P, T}}[\sigma^2_T/T] \leq 1.
		\end{equation*}
		Now, let us assume that $\rho$ is right-continuous at $P$ and $\rho(P)>0$.
		Let $\varepsilon>0$ and define
		\begin{equation*}
			P_T \coloneqq \begin{cases}
				\sqrt{1 + \frac{2\varepsilon}{|P|^2 T}} P &\text{ if } P \neq 0 \\
				\sqrt{\frac{2\varepsilon}{T}} v	&\text{ if } P = 0
			\end{cases}.
		\end{equation*}
		for $T>0$, where $v\in \R^d$ is an arbitrary unit vector.
		Then
		\begin{equation*}
			\mathcal E(|P|^2 + 2\varepsilon/T) - \mathcal E(|P|^2) = E(P_T) - E(P).
		\end{equation*}
		We write with Corollary \ref{Corollary: Point process representation of fraction of partition functions}
		\begin{equation*}
			\mathbb E_{\widehat{\Theta}_{P, T}}[e^{-\varepsilon \sigma_T^2/T}] = e^{-T(E(P_T) - E(P))} f(T) \quad \text{ where } f(T) \coloneqq \frac{\langle \Omega, e^{-T(H(P_T)-E(P_T)) } \Omega \rangle }{\langle \Omega, e^{-T(H(P)-E(P))} \Omega \rangle}
		\end{equation*}
		such that
		\begin{equation}
			\label{Equation: expression for difference quotient}
			2\varepsilon\cdot \frac{\mathcal E(|P|^2 + 2\varepsilon/T) - \mathcal E(|P|^2) }{2\varepsilon/T} = \log f(T) - \log \mathbb E_{\widehat{\Theta}_{P, T}}[e^{-\varepsilon \sigma_T^2/T}].
		\end{equation}
		By assumption
		\begin{equation*}
			\limsup_{T\to \infty} f(T) \geq \limsup_{T\to \infty} \frac{\rho(P_T)}{\rho(P) + o(T)} = 1
		\end{equation*}
		and since $\sigma^2_T/T \leq 1$ we have
		\begin{equation*}
			-\log \mathbb E_{\widehat{\Theta}_{P, T}}[e^{-\varepsilon \sigma_T^2/T}] = -\log\big(1 - \varepsilon \cdot \mathbb E_{\widehat{\Theta}_{P, T}}[\sigma_T^2/T] + \mathcal O(\varepsilon^2) \big) =  \varepsilon \cdot \hat \sigma^2_T(P) + \mathcal O(\varepsilon^2)
		\end{equation*}
		uniformly in $T$. Hence, taking the limit $T\to \infty$ in \eqref{Equation: expression for difference quotient} yields
		\begin{equation*}
			2 \varepsilon \cdot (\partial^+ \mathcal E)(|P|^2) \geq \varepsilon \limsup_{T\to \infty} \hat \sigma^2_T(P) + \mathcal O(\varepsilon^2).
		\end{equation*}
		Hence $2 \cdot (\partial^+ \mathcal E)(|P|^2) \geq \limsup_{T\to \infty} \hat \sigma^2_T(P)$ which concludes the proof.
	\end{proof}

	\begin{appendices}
	\section{Positivity improvement of the semigroup}
	Identifying $(L^2(\R^{d}))^{\otimes_s n}$ with $	L^2_{\text{sym}}(\R^{dn})$, we call $\psi = \oplus_{n} \psi_n \in \mathcal F(L^2(\R^d))$ non-negative (strictly positive) if $\psi_n \geq 0$ ($\psi_n > 0$) holds a.e.\ for all $n\in \N$. A linear operator $A$ on $\mathcal F$ is called positivity preserving (positivity improving), if $A\psi$ is non-negative (strictly positive) for all non-negative $\psi \in \operatorname{dom}(A)\setminus \{0\}$. If the semigroup generated by a lower bounded, self-adjoint operator $A$ on $\mathcal F$ is positivity improving, i.e.\ if $e^{-tA}$ is positivity improving for all $t>0$, then
	\begin{equation*}
		E_A \coloneqq \inf \sigma(A) = -\lim_{t\to \infty} \frac{1}{t} \log \langle \psi, e^{-tH(P)} \psi \rangle
	\end{equation*}
	holds for any non-negative $\psi\in \mathcal F \setminus \{0\}$ \cite[Theorem 2.2]{KellerLenzVogtWojciechowski.2015}. If $E_A$ is additionally an eigenvalue of $A$, then the corresponding eigenspace is one-dimensional and there exists a strictly positive eigenvector $\phi$, see e.g.\ \cite[Theorem 2.12]{Miyao.2010}. In particular, we then have $\langle \Omega, \phi \rangle >0$. 
	\label{Appendix: positivity improving}
	 \begin{lemma}
	 	\label{Lemma: consequences from positivity improving semigroup}
	 	The following holds for all $P\in \R^d$.
	 	\begin{enumerate}
	 		\item $E(P)$ is an eigenvalue of $H(P)$ if and only if $\rho(P)>0$.
	 		\item We have 
	 		\begin{equation*}
	 			E(P) = -\lim_{t\to \infty} \frac{1}{t} \log \langle \Omega, e^{-tH(P)} \Omega \rangle.
	 		\end{equation*}
	 	\end{enumerate}
	 \end{lemma}
	 \begin{proof}
	 While this follows similar to the proof for the Fröhlich polaron given in \cite{Miyao.2010}, we will sketch the proof for the convenience of the reader. It is sufficient to show that there exists a unitary map $U$ on $\mathcal F$ such that $U\Omega = \Omega$ and such that the semigroup generated by $U H(P) U^*$ is positivity improving. Let $\varphi:\R^d \to [0, 2\pi]$ be such that 
	 \begin{equation*}
	 	\forall k\in \R^d:\, e^{\mathrm i \varphi(k)}v(k) = -|v(k)|.
	 \end{equation*}
	 Let $V:L^2(\R^d) \to L^2(\R^d)$ be the unitary map $\psi \mapsto e^{\mathrm i \varphi} \psi$ and let $U = \Gamma(V)$. Then 
	 \begin{equation*}
	 	U H(P) U^* =  \hat H(P) \quad \text{ where }\quad  \hat H(P) = \frac12 |P - P_f|^2 + \mathrm d\Gamma(\omega) + \phi(-|v|).
	 \end{equation*}
	 We hence might assume with out loss of generality that $v\leq 0$.
	 For $n\in \mathbb N$, we define the regularized Hamiltonian $ H_n(P)$ by
	 	\begin{equation*}
	 		 H_n(P) = \frac12 |P - P_f|^2 + \mathrm d\Gamma(\omega) + \phi(v_n), \quad \text{ where }v_n(k) \coloneqq e^{-|k|/n}v(k).
	 	\end{equation*}
	 	Notice that our assumptions imply $v_n \in L^2(\R^d)$ and $v_n(k)<0$ for all $k\in \R^d \setminus\{0\}$ and $n\in \N$.
	 	Combining \cite[Appendix A]{Desio.Seiringer.2025} and \cite[Chapter VIII, Theorem 3.11]{Kato.1980}, one obtains $H_n(P) \to  H(P)$ in the strong resolvent sense. One easily checks that $ H_n(P) -  H_m(P)$ is positivity preserving for all $m\geq n$. As in the proof of \cite[Proposition 6.8]{Miyao.2010} one can show that the semigroup generated by $H_n(P)$ is for every $n\in \mathbb N$ positivity improving. By applying \cite[Theorem 2.4]{Miyao.2013}, we obtain that the semigroup generated by $ H(P)$ is positivity improving.
	 \end{proof}
	\end{appendices}

\end{document}